\documentclass[a4paper,12pt]{article}
\usepackage[utf8]{inputenc}
\usepackage{amsmath,amsthm,amsfonts,amssymb}
\usepackage{a4wide}
\usepackage{xspace}
\usepackage{xfrac}
\usepackage{xcolor}
\usepackage{graphicx}
\usepackage{subfig}
\usepackage{bm}

\usepackage[round]{natbib}

\newcommand{\e}{\mathrm{e}}
\newcommand{\rd}{\mathrm{d}}
\newcommand{\R}{\mathbb{R}}

\newcommand{\bF}{\mathbf{F}}
\newcommand{\bM}{\mathbf{M}}
\newcommand{\bp}{\mathbf{p}}
\newcommand{\A}{\ensuremath{\mathbb{A}}\xspace}
\newcommand{\B}{\ensuremath{\mathbb{B}}\xspace}

\newcommand{\erf}{\mathrm{erf}}

\newtheorem{lemma}{Lemma}

\newtheorem{remark}{Remark}
\newtheorem{theorem}{Theorem}
\newtheorem{definition}{Definition}

\newtheorem{problem}{Problem}

 \graphicspath{{./Figures/}}

\title{From fixation probabilities to $d$-player games: an inverse problem in evolutionary dynamics}
\author{Fabio A. C. C. Chalub\thanks{Departamento de Matemática and Centro de Matemática e Aplicações, Faculdade de Ciências e Tecnologia, Universidade Nova de Lisboa, Quinta da Torre, 2829-516, Caparica, Portugal.} and Max O. Souza\thanks{Instituto de Matemática e Estatística, Universidade Federal Fluminense, R. Prof. Marcos Waldemar de Freitas Reis, s/n, 24210-201, Niter\'oi, RJ, Brasil.}}

\begin{document}

\maketitle

\begin{abstract}
The probability that the frequency of a particular trait will eventually become unity, the so-called fixation probability, is a central issue in the  study of population evolution. Its computation, once we are given a stochastic finite population model without mutations and a (possibly frequency dependent) fitness function, is straightforward and it can be done in several ways. Nevertheless, despite the fact that the fixation probability is an important macroscopic property of the population, its  precise knowledge does not give any clear information about the interaction patterns among individuals in the population. Here we address the inverse problem: From a given fixation pattern and population size, we want to infer what is the game being played by the population. This is done by first exploiting the framework developed in FACC Chalub and MO Souza, J. Math. Biol. 75: 1735, 2017., which yields a  fitness function that realises this fixation pattern in the Wright-Fisher model. This  fitness function always exists, but it is not necessarily unique. Subsequently, we show that any such  fitness function can be approximated, with arbitrary precision, using $d$-player game theory, provided $d$ is large enough. The pay-off matrix that emerges naturally from the approximating game will provide useful information about the individual interaction structure that is not itself apparent in the fixation pattern. We present extensive numerical support for our conclusions.
\end{abstract}

\section{Introduction}

Evolutionary Game Theory (EGT) was introduced in the earlier 70's of the twentieth century in the seminal work by  \cite{smith1973logic}, as a convenient and useful tool to describe animal interactions, where  no rational behaviour can be assumed --- in contradistinction  to its traditional application in economics. The first achievement of EGT was the derivation of mixed  stable evolutionary states in infinite populations~\citep{JMS}. Subsequently, the replicator equation was introduced in evolutionary dynamics, using the so called fitness, to model relative growth of biological populations \citep{taylor1978evolutionary}. Assuming that individuals in a given population interact according to game theory, their fitnesses were obtained from the pay-offs of the underlying games; see~\cite{HofbauerSigmund1998} and references therein.

On the other hand,  the much older Wright-Fisher (WF) process was introduced in the late 30's as a model for the description of the evolution of gene frequencies in finite populations subdivided into finitely many different groups, homogeneous in all characteristics but one under study~\citep{Fisher1,Wright1,Wright2}. One of the fundamental features of the WF process, as considered here and in several other references~\citep{Ewens_2004,Burger_2000}, is that it includes genetic drift, but not mutation.  Hence, if the population reaches a homogeneous state, i.e. a state with all individuals being of a single type,  then it will remain in this way forever--- in Markov chain theory parlance these are termed absorbing states \citep{Karlin_Taylor_intro}.

Finite population processes with absorbing states share one important characteristic: the dynamics will eventually reach an absorbing state. In the case of the WF process, this means that the population will become homogeneous in the long run. This behaviour is important when modelling situations in time scales that are faster than the mutational one --- so that  one expects  selection and genetic drift to have enough time to act. In particular, the latter will be the ultimate cause for one type to  fixate, i.e. to  displace all the other types before a new mutation takes place. In this framework, an important question is the following: given the current state of the population, what is probability that a certain type will fixate?  This is considered one of the central questions in the mathematical study of evolution. It is worth noting that, while there is no explicit formula for the fixation probability of the finite population WF process (except for the so called neutral evolution), its numerical determination can be easily done using  finite Markov chains and/or numerical linear algebra techniques~\citep{Karlin_Taylor_intro}.

Initial work with the WF process assumed constant fitnesses. This assumption has been already relaxed, in the realm of mathematical population genetics, as early as in \citep{EthierKurtz}. However, it was just in the last decade that EGT and WF processes were first combined as important modelling tools for studying the evolution of a population~\citep{Imhof:Nowak:2006}. 
 
Before moving on, we should point out  that another important finite population process is the Moran process. 
On one hand, following~\citet{Hartle_Clark}, we define the Darwinian fitness as the ratio of the expected number of individuals of a given type between successive generations. On the other hand, in the classical definition of frequency dependent Moran process, the probability to be selected for reproduction is assumed to be proportional to the type fitness~\citep{Nowak:06}. A simple calculation shows that these concepts are not, in general, compatible, cf.~\cite{ChalubSouza:2017a}.

From now on, we will restrict ourselves to finite populations with two types, but we will not make any \emph{a priori} assumption about their interaction. Indeed, we will show that almost any arbitrary interaction can be modelled by $d$-player game theory, provided we do not impose any restriction on the game size. As  already observed \citep{pacheco2009evolutionary,souza2009evolution,wu2013dynamic,pena2014gains,ChalubSouza16,ChalubSouza:2017a,czuppon2018disentangling,ChalubSouza18}, $d$-player games with $d\geq3$ can exhibit quite distinct behaviour from the 2-player ones. In the multi-type case, we expect that a similar approach might work, but much remains to be done---indeed, multi-type evolution is an important source of open mathematical problems.

 Initially, as mentioned above, only 2-player game theory was used to describe the evolution of a population in the non-constant fitness case, also called ``frequency dependent fitness''; in that case the fitness is an affine function with respect to the fraction of different individuals in the population; see, however,~\citet{Kurokawa:Ihara:2009,Gokhale:Traulsen:2010,Gokhale:Traulsen:2014} for the recent use of $d$-player game theory in the description of evolutionary processes. 
 
 The authors have been using arbitrary fitness functions, i.e. fitness functions not necessarily derived from EGT, for almost a decade now when studying continuous approximations--- cf. \citet{ChalubSouza09b,ChalubSouza14a,ChalubSouza16}. However, it was only when studying finite populations that the chasm became more evident. In \cite{ChalubSouza:2017a}, it was shown that fixation patterns of the WF process derived from 3-player game theory can be fundamentally distinct from their 2-player game theory counterpart. In  the latter framework, fixation probability is always an increasing function of the initial presence, whereas in the former, or in more general frameworks, this is not necessary the case. Non-increasing fixation probabilities seem to be a topic not fully explored in the evolutionary dynamic literature; we however conjecture --- at this point, without any experimental base -- in~\cite{ChalubSouza:2017a} that this may be related to missing record fossils. In fact, if the fixation probability is a non increasing function of the focal type initial presence, then an homogeneous state is more likely to be attained from a \emph{jump} from a mixed state than trough a slow process of accumulation. Finally, we point out that non monotonicity of fixation might also appear in 2-player games, due to the stochastic fluctuations in population size  --- cf. \cite{huang2015stochastic} and \cite{czuppon2018fixation}, but without explicitly mentioning it.
 
 From a practical point of view, once the fitness functions are chosen the fixation probability is well determined and, indeed, it can be easily computed numerically. The present work is devoted, however, to the inverse problem: to which extent does the fixation probability determine the fitnesses of all different types in the WF evolution, and how  can the fitnesses functions be reasonably approximated using $d$-player EGT, for $d$ large enough.

 A partial solution was presented in~\cite{ChalubSouza:2017a}, where it was shown
  that any non-degenerated fixation probability vector can be obtained as the result of evolution by the WF process, with selection given by a certain discrete relative fitness function; furthermore this relative fitness function is uniquely determined (modulo some trivial transformations) if and only if the fixation probability vector is increasing as a function of the initial presence of the focal type. 
 
 In the present work, we will go one step further and show that any fitness can be arbitrarily well approximated by pay-offs obtained from finite-player game theory. Therefore, we will be able to show how any fixation pattern can be approximated using $d$-player game theory. In particular, this shows, or at least indicates, the minimum number of individuals required in each relevant interaction that are necessary to conveniently describe the long run evolution of a WF population. To the best of our knowledge, this is the only work  where microscopic interactions among individuals in a population can be inferred from the results of long term evolution, in particular, from the fixation probability. Fig.~\ref{fig:direct_inverse} summarizes the direct problem in WF dynamics, and the inverse problem that is studied in this work.
 
 \begin{figure}
  \centering
  \includegraphics[width=0.48\textwidth]{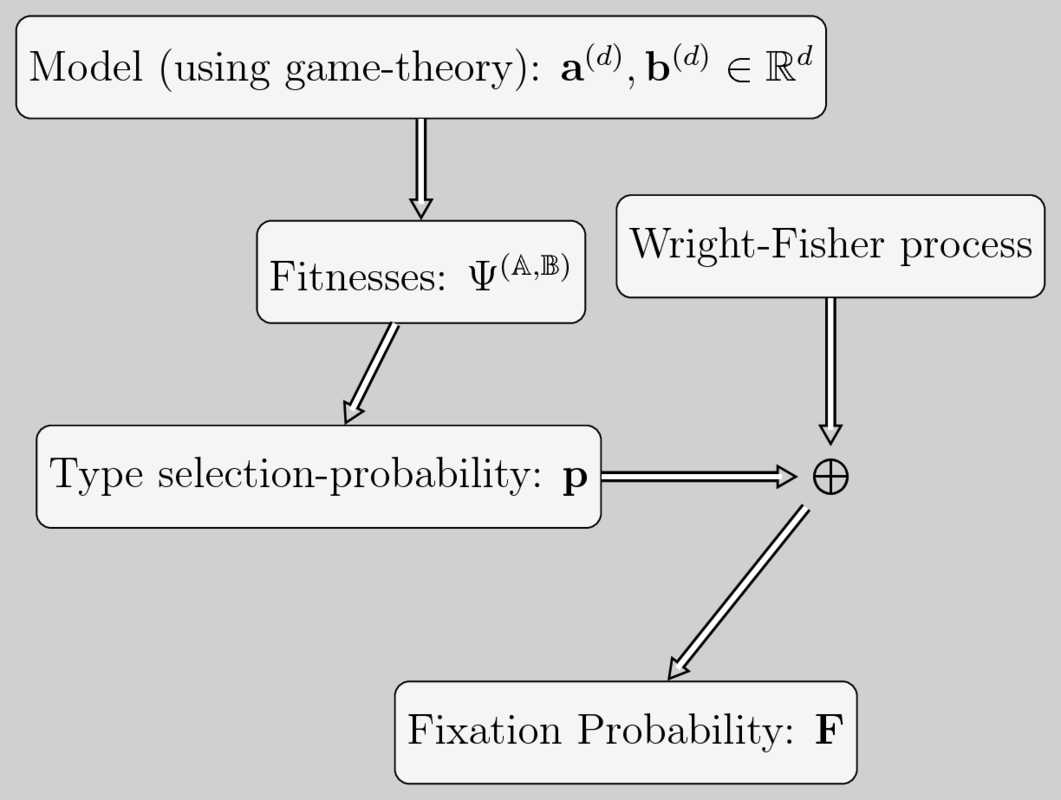}
  \includegraphics[width=0.48\textwidth]{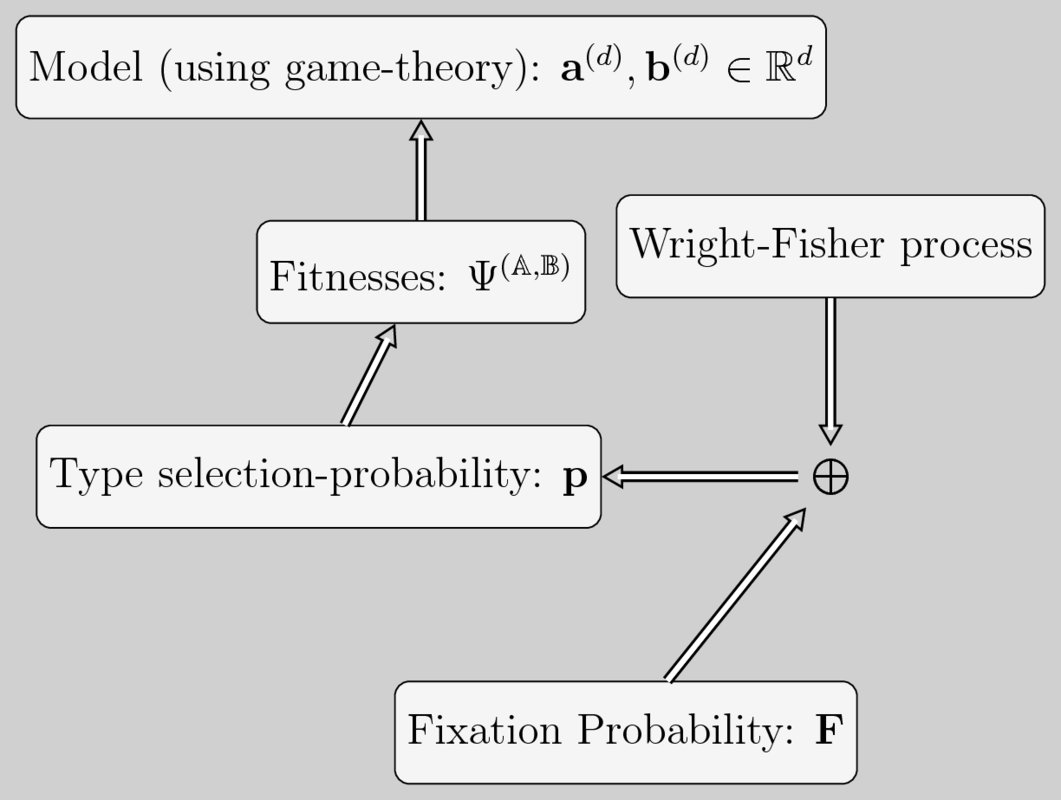}
  \caption{Direct (left) and inverse (right) problem in the WF dynamics. In the former case, from an assumed interaction between individuals, we proceed step-by-step until finding the fixation probability for all initial conditions. In the latter case, for a given $d$ and fixation, we find the $d$-player game that better approximates the given fixation.}
  \label{fig:direct_inverse}
 \end{figure}

\section{Fitness, fixation and the Wright-Fisher process}
\label{sec:prelim}

In this Section,  we will  review a number of definitions and results on the Wright-Fisher (WF) process that will be relevant in the sequel. This  review will follow closely  \cite{ChalubSouza:2017a} --- in particular, we  will consider WF processes for two types, \A and \B, that are parametrised by \textit{type selection probabilities} (TSP). More precisely, we will say that the population is at state $j$, when there are $j$ individuals of type \A. We will write $\bM=\left(M_{ij}\right)_{i,j=0,\dots,N}$ for the process transition matrix --- the  probability to go from state $j$ to $i$ ---  and $\bp=\left(p_i\right)_{i=0,\ldots,N} \in \lbrace0\rbrace\times [0,1]^{N-1}\times\lbrace1\rbrace$ for the vector of TSP --- the probability that an individual of  type \A is chosen for reproduction in a population with $i$ individuals of this type. 

In this framework, we have that
\[
M_{ij}=\binom{N}{i}p_j^i(1-p_j)^{N-i}\ .
\]
\begin{remark}
	We will follow the convention in  \cite{ChalubSouza:2017a}, and will index  vectors and matrices from $0$ to $N$, unless otherwise stated. On the other hand, $\bM$ as written above is column stochastic rather than being row stochastic as in \cite{ChalubSouza:2017a}.
\end{remark}

\begin{definition}
	\label{def:cons_lem}
	A stochastic population process $X_k$ is consistently parametrised  if, and only if, we have that
	\[
	\mathbb{E}[X_{k+1}|X_k=j]=Np_j\ \forall k.
	\]
\end{definition}

\begin{remark}
	A process is consistently parametrised, if for any given population state, the probability to select the focal type to reproduce at time $t$ is equal to the expected fraction of individuals of the focal type in the next generation. In particular 
	\[
	\Phi(j):=\frac{\mathbb{E}[X_{k+1}|X_k=j]/j}{(N-\mathbb{E}[X_{k+1}|X_k=j])/(N-j)}=\frac{N-j}{j}\frac{p_j}{1-p_j} 
	\]
	is such that 
	\[
	p_j=\frac{j\Phi(j)}{j\Phi(j)+N-j}\qquad j=1,\dots,N-1
	\]
	and therefore, $\Phi(j)$, the so called \emph{Darwinian fitness}, is directly identified as a proxy to the reproductive probability. As a consequence of Lemma~1 in \citet{ChalubSouza:2017a}, we conclude that the WF process is consistently parametrised.
\end{remark}

\begin{definition}
	\label{def:adm_vfp}
	We will say that $\bF\in [0,1]^{N+1}$ is an admissible vector of fixation probabilities, if  $\bF\in \lbrace0\rbrace\times (0,1)^{N-1}\times\lbrace1\rbrace\subset \R^{N+1}$. Also, we will say that a vector $\mathbf{a}\in\R^n$ is \emph{strictly increasing} (non-decreasing) if $a_i<a_{i+1}$ ($a_i\le a_{i+1}$, respec.) for $i=0,\dots, n-1$. Finally, an admissible fixation vector will be termed: \textsf{regular}, if it is strictly increasing and \textsf{weakly-regular} if it is non-decreasing. Otherwise, it will be termed non-regular. 
\end{definition}

\begin{theorem}
	\label{thm:wf_fixation}
Every vector of fixation probabilities of a WF process is admissible. Conversely, given an admissible vector of fixation probabilities $\bF$, then there exists at least one WF process, such that $\bF$ is the associated fixation vector. In addition, if $\bF$ is regular, then this process is unique.   
\end{theorem}

\begin{remark}
	Definition~\ref{def:adm_vfp} is a compacted version of Definitions~4 and~5 in \citet{ChalubSouza:2017a}, while the first part of Theorem~\ref{thm:wf_fixation} is contained in Proposition~1 and the second part is  Theorem~5 of \citet{ChalubSouza:2017a}. WF processes with regular fixation were, by extension,  termed regular in \citet{ChalubSouza:2017a}. Theorem~\ref{thm:wf_fixation} has an additional consequence that is of interest: regular WF processes are uniquely characterised by their fixation. In particular, all their transient properties are fully determined by their fixation vector.
\end{remark}

In what follows, we will need to recall the auxiliary results used to prove the second part of Theorem~\ref{thm:wf_fixation}. 

\begin{lemma}
	\label{lem:upsilon}
	Let $\bF$ be an admissible fixation vector, and consider 
	\begin{equation}
	\label{eq:upsilon}
	\Upsilon_{\bF}(p)=\sum_{i=0}^NF_i\binom{N}{i}p^i(1-p)^{N-i}\ .
	\end{equation}
	Then the following holds true:
	\begin{enumerate}
		\item $\Upsilon_{\bF}(0)=F_0=0$, $\Upsilon_{\bF}(1)=F_N=1$, and for $p\in(0,1)$, $\Upsilon_{\bF}(p)\in(0,1)$.
		\item $\Upsilon_{\bF}:[0,1]\to[0,1]$ is continuous and onto.
		\item If $\bF$ is regular, then $\Upsilon_{\bF}$ is an increasing function.
		\item Let $\bp$ be a TSP vector and consider the associated WF process that will denote by $\bM(\bp)$. Then $\bF$ is the corresponding fixation probability vector if, and only if, we have
		\[
		\Upsilon_{\bF}(p_j)=F_j,\quad j=1,\ldots,N-1.
		\]
	\end{enumerate}
\end{lemma}

\begin{remark}
	Lemma~\ref{lem:upsilon} is a collection of results that are somewhat scattered in \citet{ChalubSouza:2017a}. The first part of (1) and (2) are easily checked --- admissibility comes into play only in the second part of (1). (3) is a well-known property of Bernstein polynomials --- cf.~\citet{Phillips:2003,Gzyl2003}.  (4) is the actual reason why $\Upsilon_{\bF}$ is an important object, and the easiest way to verify it is to check, from the definitions, that  $\left(\bF^T\bM(\bp)\right)_j=\Upsilon_{\bF}(p_j)$.
\end{remark}

\section{The inverse problem for the fitness}

\subsection{The problem setup}

In Section we will elaborate on the inversion of  fixation into relative fitness, and provide some additional results

\begin{problem}[Inverse problem for fitness]
	\label{probl:ipf}
	Let us assume that we are either  given a vector of admissible fixation of probabilities $\bF\in\R^{N+1}$ or a fixation function $f:[0,1]\to[0,1]$ satisfying(i) $f(0)=0$; (ii) $f(1)=1$; (iii) $0<f(x)<1$, $x\in(0,1)$, together with a population size $N$, such that $F_j=f(j/N)$.
	The inverse problem for fitness is to find $\Phi(j)$, $j-1,\ldots,N-1$, such that, if we let
	\[
	p_j=\frac{j\Phi(j)}{j\Phi(j)+N-j},\quad j=1,\ldots,N-1,
	\]
	then we have
	\[
	\bF^T\bM(\bp)=\bF^T,
	\]
\end{problem}
where $\bM(\bp)$ is the WF matrix for TSP $\bp$.
In view of Theorem~\ref{thm:wf_fixation}, Problem~\ref{probl:ipf} always has an exact solution. The next Lemma will be helpful when dealing with the inverse problem in a more concrete way:
\begin{lemma}
	\label{lem:fix2relfit}
	For $y\in[0,1]$, let $\Upsilon_{\bF}^{-1}(y)=\lbrace x\in[0,1]\,:\, \Upsilon_{\bF}(x)=y\rbrace$. Let $\bF'\in [-1,1]^N$ be defined by $F'_i=F_{i+1}-F_i$, $i=0,\ldots,N$, and define $\mathrm{sgn}(\bF')$ as the number of sign changes in $\bF'$.  Then, for all $y\in[0,1]$ we have that $1\leq \#\Upsilon_{\bF}^{-1}(y)\leq 1+\mathrm{sgn}(\bF')$. In particular, it follows that we can always define the inverse map of $\Upsilon_{\bF}$ as a set-valued map, and hence we have
	\begin{equation}
	\label{eq:fix2relfit}
	\Phi(j)=\frac{N-j}{j}\frac{\Upsilon_{\bF}^{-1}(F_j)}{1-\Upsilon_{\bF}^{-1}(F_j)}
	\end{equation}
\end{lemma}
\begin{proof}
	The lower bound comes from the fact that $\Upsilon_{\bF}$ is onto. The upper bound comes from the fact that we can see $\Upsilon_{\bF}$ as a two-dimensional  Bézier curve and $\bF$ as its control polygon. Hence,  it satisfies the so-called total variation diminishing property \citep{lane1983geometric} and it must oscillate less than its control polygon. The number of zeros of a function is  bounded above by the number of oscillations plus one. 
\end{proof}
From a theoretical viewpoint, Equation~\ref{eq:fix2relfit} encodes all the information  on the mapping from  fixation probabilities into relative fitnesses. 
\begin{definition}
	The critical set of $\Upsilon_{\bF}$ is given by
	\[
\mathrm{Crit}\left(\Upsilon_{\bF}\right)=\bigcup_{y\,:\, \Upsilon_{\bF}'(y)=0} \Upsilon_{\bF}^{-1}(y).
\]
Since $\Upsilon_{\bF}$ is a polynomial, its critical set is always finite. 
\end{definition}
We are now ready to discuss the cases:
\begin{description}
\item[$\Upsilon_{\bF}$ is increasing]: In this case, $\Upsilon_{\bF}^{-1}$ is well defined and smooth, except over $\mathrm{Crit}\left(\Upsilon_{\bF}\right)$ where it will be only continuous.
\item [$\Upsilon_{\bF}$ is not increasing]: In this case, $\Upsilon_{\bF}^{-1}$ is defined only as set-valued function and further work is necessary. 
\end{description}
In order to deal with the non-monotonic case, we begin with a definition:
\begin{definition}
	A \textsl{branch} of a set-valued function is a choice of a single element in each of its set-values, such that $\Upsilon_{\bF}^{-1}$ is well defined and smooth, except over $\mathrm{Crit}\left(\Upsilon_{\bF}\right)$,  where it  will have to be discontinuous  at least over a point. 
\end{definition}

\begin{remark}
Branches are not unique, and we will single out two of them:
\[
\Upsilon_{\bF}^{-1,+}(y)= \max \Upsilon_{\bF}^{-1}(y),
\]
with a similar definition using $\min$ for $\Upsilon_{\bF}^{-1,-}$.
\end{remark}

We now present a result on the stability of this recovering:

\begin{lemma}
	Let $\Phi(j)$ be a solution for input $F_j$ (or the solution, in the increasing case), and let $\Phi(j)+\Delta\Phi(j)$ be the solution for input $F_j+\Delta F_j$. Then
		\begin{equation}
	\label{eq:error_bars_new}
	\Upsilon_{\bF}'\left(p_j\right)(1-p_j)^2\frac{j}{N-j}\Delta\Phi(j)=\Delta F_j,\quad j=1,\ldots,N-1.
	\end{equation}
\end{lemma}

\begin{proof}
Let
	\[
	p_j+\Delta p_j=\frac{j\left(\Phi(j)+\Delta\Phi(j)\right)}{j\left(\Phi(j)+\Delta\Phi(j)\right)+N-j}.
	\]
	Then, we obtain that
	\[
	\Delta p_j=\frac{j(N-j)}{\left(N+j\left(\Phi(j)-1\right)\right)^2}\Delta\Phi(j)+\mathcal{O}\left(\left(\Delta\Phi(j)\right)^2\right).
	\]
	Similarly, from the identity
	\[
	\Upsilon_{\bF}\left(p_j+\Delta p_j\right)=F_j+\Delta F_j,\quad \Upsilon_{\bF}(p_j)=F_j,
	\]
	we obtain that
	\[
	\Upsilon_{\bF}'\left(p_j\right)\Delta p_j + \mathcal{O}\left(\left(\Delta p_j\right)^2\right) =\Delta F_j
	\]
	These two expressions yield, dropping higher order terms in $\Delta\Phi(j)$, that
	\begin{equation*}
	\Upsilon_{\bF}'\left(p_j\right)\frac{j(N-j)}{\left(j\Phi(j)+N-j\right)^2}\Delta\Phi(j)=\Delta F_j
	\end{equation*}
	Using the definition of $\Phi(j)$ in terms o $p_j$ we finally get Equation~\eqref{eq:error_bars_new}.
\end{proof}

\begin{remark}
	Equation~\ref{eq:error_bars_new} shows that, if $\Upsilon_{\bF}'>\alpha>0$, then  
	\[
	K=\left(\alpha\min_{j=1,\ldots,N-1}(1-p_j)^2\frac{j}{N-j}\right)^{-1}
	\]
	is an upper-bound on the amplification factor. Naturally, $K$ can b very large if $\Upsilon_{\bF}$ has a plateau-like form, at least over some subinterval, which might lead to accuracy loss of the inversion. It also shows that, when $\Upsilon_{\bF}$ is non-monotonic, then the inversion might also lose precision, which was already hinted by the non-continuity of the branches along the critical set of $\Upsilon_{\bF}$.
\end{remark}
	  
\subsection{Computational aspects}
\paragraph{Increasing $\Upsilon_{\bF}$:}

In this case, $\Upsilon_{\bF}^{-1}$ is single valued, and hence Equation~\ref{eq:fix2relfit} is  uniquely defined as function of $F_j$. From a numerical viewpoint, we let $Q_j(p)=\Upsilon_{\bF}(p)-F_j$ and solve it either using a good shelve root finder or using a more specific code. The former should work in most cases, but might fail for fixation probabilities the have a plateau-like behaviour --- cf. \cite{ChalubSouza16}. The latter can be obtained by using a stable numerical evaluation of $\Upsilon_{\bF}$ by de Casteljeau's algorithm together with a combination  of  bisection and continuation methods --- cf. \citep{press1992numerical} and \citep{allgower2012numerical}.

\paragraph{Non-increasing $\Upsilon_{\bF}$:}

In this case, Equation~\eqref{eq:fix2relfit} yields a set-valued map, and even if we assume, for the time being, that we are able to approximate all set values accurately, we are still left with the choice of elements in the set values that are not singletons. Here there are two main avenues: (i) since Equation~\eqref{eq:fix2relfit} is only to be enforced on a finite number of points, we can view it as a discrete problem and then consider some subset of all possible combinations --- e.g. taking the minimum or maximum from each set value; (ii) consider it as continuous problem, and then choose a single-valued version (SVV) for $\Upsilon_{\bF}^{-1}$ that is sufficiently regular. Indeed, since $\Upsilon_{\bF}$ is locally invertible by the Inverse Function Theorem outside its critical points, we can always  restrict the discontinuities  to the inverse images of the critical points of $\Upsilon_{\bF}$, and obtain  SVVs that are semicontinuous (either lower or upper). It should be noticed that these approaches are not necessarily incompatible: for instance, taking always the minimum or the maximum yields such regular SVVs --- the former is lower semicontinuous, while the latter is  upper semicontinuous.

We now turn back our attention to the numerical approximation of the set values of $\Upsilon_{\bF}^{-1}$. In order to do this, we need to approximate all the roots for each $Q_j$. For this task, there are many algorithms that allow one to isolate the roots of a polynomial in a given interval as the Vincent-Akritas-Strzebo\'nski (VAS) method or  Sturm sequences or Boudan-Fourier theorem combined with the BN algorithm --- cf. \citep{rouillier2004efficient,alesina2000vincent}. Somewhat like in the previous case, these algorithms might fail when working with fixation patterns that mix oscillatory-like with plateau-like regions.

In what follows, we will assume that Equation~\eqref{eq:fix2relfit} has been inverted by a shelve numerical method and hence that $\Phi(j)$ is available for $j=1,\ldots,N-1$.

\section{Games from given relative fitnesses vectors}
\label{sec:relfit2games}

Our aim now is to identify when a fitness recovered by the procedure described in the previous Section can be explained by EGT.

\subsection{Problem setup}

Recall \citep{Gokhale:Traulsen:2010,Kurokawa:Ihara:2009,Lessard:2011} that in a two-strategy, symmetric $d$-player game,  each individual in a  population of size $N$  interacts at once with $d-1$ other individuals; if there are $k$ type \A individuals among the $d-1$ opponents of the focal individual, then her pay-off is $a_k$ (if she is of type \A) or $b_k$ (if she is of type \B).  Therefore, a symmetric two-strategy $d$-player game is defined by two vectors $\mathbf{a}^{(d)},\mathbf{b}^{(d)}\in\R^{d}$. 

Since  every group of $d$ interacting individuals is chosen at random from the entire population, the average pay-off of type \A and \B individuals is given by the hyper-geometric distribution. Thus, we have for $j=1,\ldots,N-1$ that the average pay-off is 
\begin{equation}
\label{eq:dpp}
\varphi_\A(j)=\sum_{k=0}^{d-1}\frac{\binom{j-1}{k}\binom{N-j}{d-1-k}}{\binom{N-1}{d-1}}a_k^{(d)}
\quad\text{and}\quad
\varphi_\B(j)=\sum_{k=0}^{d-1}\frac{\binom{j}{k}\binom{N-j-1}{d-1-k}}{\binom{N-1}{d-1}}b_k^{(d)}.
\end{equation}
 we need to solve for the coefficients $\mathbf{a}^{(d)}$ and $\mathbf{b}^{(d)}$ the following $N-1$ equations:
\begin{equation}
\label{eq:fit_gam}
\Phi(j)=\frac{\varphi_\A\left(j\right)}{\varphi_\B\left(j\right)},\quad j=1,\ldots,N-1.
\end{equation}
Before we proceed any further, we want to point out two issues:

\begin{enumerate}
	\item If $\lambda\in\R$, then the games defined by the pay-offs  $(\mathbf{a}^{(d)},\mathbf{b}^{(d)})$ and  $(\lambda\mathbf{a}^{(d)},\lambda\mathbf{b}^{(d)})$ are  associated to the same relative fitness given by Eq.~\eqref{eq:fit_gam} --- hence, a particular  solution can be identified with a line  through the origin in $\R^{2d}$, and we  cannot expect uniqueness, unless some normalisation is chosen.
	\item Unless $d$ is sufficiently large, we cannot expect the system given in Equations~\eqref{eq:dpp} and~\eqref{eq:fit_gam} to have solutions at all.
\end{enumerate}
We will now elaborate on (2) above: Let $\bm{\Phi}=(\Phi(1),\ldots,\Phi(N-1))^T$, and $\bm{A}$ and $\bm{B}$ be $(N-1)\times d$ matrices defined by
\[
A_{jk}=\binom{j-1}{k}\binom{N-j}{d-1-k}\text{ and }B_{jk}=\binom{j}{k}\binom{N-j-1}{d-1-k},
\]
for $j=1,\ldots,N-1\text{ and }k=0,\ldots,d-1$. Also, let Let $\bm{x}=(\bm{a}^{(d)},\bm{b}^{(d)})^T$ and let $\bm{T}$ be a $(N-1)\times (2d)$ matrix such that
$T_{ij}=A_{ij}$ for $0\le j\le d-1$ and $T_{ij}=-\left(\mathrm{diag}(\bm{\Phi})\bm{B}\right)_{i,j-d}$ for $d\le j\le 2d-1$,
where $\mathrm{diag}(\bm{v})$ is the diagonal matrix whose main diagonal is given by $\bm{v}$. Then, Equations~\eqref{eq:dpp} and~\eqref{eq:fit_gam} implies $\bm{T}\bm{x}=0$, with $\bm{x}\not=0$. Conversely, if $\bm{x}\in \mathrm{Null}(\bm{T})$, and if $\bm{A}\;\bm{a}^{(d)}$ and $\bm{B}\;\bm{b}^{(d)}$ do not have null entries (in this case, the  corresponding entry vanishes for both of them), then  the entries of $\bm{x}$ are a solution to Equations~\eqref{eq:dpp} and~\eqref{eq:fit_gam}. In particular, if $d=\sfrac{N}{2}$ (for even $N$) or if $d=\sfrac{(N+1)}{2}$ (for odd $N$), then $\bm{T}$ must have a non-trivial null space, and if the extra condition holds we have a solution.

From now on, we give-up on have an exact solution and pose the following problem
\begin{problem}[Inverse problem for $d$-player games]
	\label{probl:ig}
	Given a fitness vector $\bm{\Phi}$, and $2\leq d \in \mathbb{N}$, find coefficients $\bm{a}^{(d)}$ and $\bm{b}^{(d)}$ that solve Equations~\eqref{eq:dpp} and~\eqref{eq:fit_gam} in least squares sense.
\end{problem}
Problem~\ref{probl:ig} can be recast in two ways:
\begin{enumerate}
	\item Let  $F(\bm{x})= \bm{A}\;\bm{a}^{(d)} \oslash\bm{B}\;\bm{b}^{(d)}$, where $\oslash$ denotes Hadamard division. Then, the corresponding non-linear least squares problem is given by
	\[
	\bm{x}_*=(\bm{a}_*^{(d)},\bm{b}_*^{(d)})^T\in\arg\min \|\bm{\Phi}-F(\bm{x})\|_2^2.
	\]
	\item Consider the matrix $\bm{T}$ defined above, and let $\bm{T}=\bm{U}\bm{\Sigma}\bm{V}^*$ be the SVD decomposition of $\bm{T}$, where $\bm{U}$ is $(N-1)\times(N-1)$, $\bm{\Sigma}$ is $(N-1)\times(2d)$ and $\bm{V}$ is $(2d)\times(2d)$. Then, a tentative solution is given by the last column of $\bm{V}$ --- provided it satisfies the additional condition.
\end{enumerate}

The whole procedure is summarised in Fig.~\ref{fig:direct_inverse}, noting that the map from fixation into fitness may be multivalued -- in this case, it is necessary to select a branch.

\subsection{Acceptability of a given game}
\label{sec:good_games}

For a positive vector $\bm{\omega} \in \R^{N-1}$, and a vector $\mathbf{x}\in\R^{N-1}$,  the weighted $L_\infty$ norm is defined by
\[
\|\mathbf{x}\|_{\infty,\bm{\omega}}=\max_{i=1,\ldots,N-1}|x_i|\omega_i.
\] 
For a given fixation $\bF$ and population size $N$, 
we will deem  the approximation  $\widehat{\bF}$ produced by the procedure above as \textit{a good approximation} or (\textit{acceptable}), if 

\begin{equation}
\label{eq:max_error_new}
\|\bF-\widehat{\bF}\|_{\infty,\omega}\leq \epsilon_N,
\end{equation}
where $\epsilon_N$ is the maximum tolerable error when  population size is $N$.  The use of a weighted norm allows to choose different thresholds for the different observed initial frequencies. The choice of  weight and tolerance constants are,  to a large extent,  arbitrary. In what follows we will consider $\bm{\omega}=\mathbf{1}$, and $\epsilon_N=0.01$. However,  a second natural choice would be to let $\omega_i=\left(\sfrac{i}{N}\left(1-\sfrac{i}{N}\right)\right)^{-1/2}$, and $\epsilon_N=\sfrac{\kappa}{\sqrt{N}}$, i.e., a maximum tolerable error proportional to the natural fluctuations of the stochastic system.

Let us call $\mathcal{A}_F$ the set of acceptable  game sizes, i.e., if $d\in\mathcal{A}_F$, then there exists a  fixation produced by a $d$-player game, which is agood approximation. Note that, since $N\in\mathcal{A}_F$, we have that $\mathcal{A}_F\not=\emptyset$ and that, if $d\in\mathcal{A}_F$, then $d'\in\mathcal{A}_F$ for all $d'>d$.  

\begin{definition}
Given a fixation vector $\bF$, we define its  \textsf{complexity} as 
\[
d_{\min}(F)=\min_{d'\in\mathcal{A}_F} d'.
\]
\end{definition}
In other words, the fixation  complexity is the minimum number of players necessary in a typical micro-interaction to reproduce this fixation pattern of the Wright-Fisher dynamics,  with an error not greater than the error prescribed by Equation~\eqref{eq:max_error_new}.

\section{Examples}

Let us consider a fixed population size $N$ and assume that the fixation vector $F_j$, $j=1,\dots,N-1$, $F_0=1-F_N=0$, can be obtained from sampling a continuous function $f:[0,1]\to[0,1]$, with $f(0)=0$ and $f(1)=1$, on the grid of population frequencies. Namely, we assume that $F_j=f\left(\frac{j}{N}\right)$. We will call $f$ \emph{the fixation function}.

We start by studying examples in which $f$ is an increasing function in the interval $[0,1]$; in this case  the sequence $\left(p_i\right)_{i=0,\dots,N}$ is  uniquely determined and increasing.

In Figs.~\ref{fig:dominance}, \ref{fig:coexistence}, and \ref{fig:coordination} we study the three most relevant cases inherited from two-player evolutionary game theory  studies: dominance, coexistence and coordination, respectively. However,  we consider fixation functions $f:[0,1]\to[0,1]$, with $f(0)=1-f(1)=0$, that, despite the fact that they do not arise directly from 2-player game theory, are such that their comparison to the neutral case as well as their concavity in the interval $[0,1$] make them directly comparable with the classical 2-player case, as described above. As mentioned in Section~\ref{sec:good_games}, we will take $\bm{\omega}=\mathbf{1}$ and $\epsilon_N=0.01$.

\begin{figure}
\centering
\includegraphics[width=0.47\textwidth]{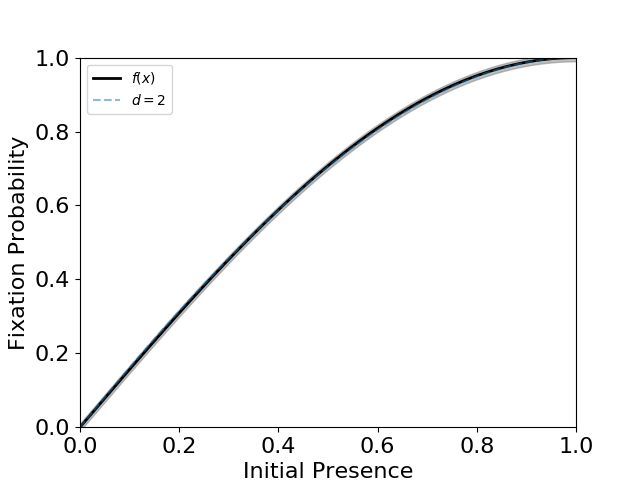}
\includegraphics[width=0.47\textwidth]{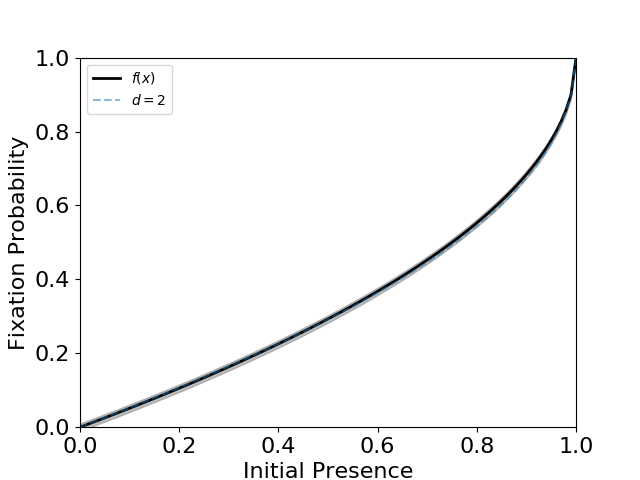}
\caption{Dominance case. Left: Fixation probability is obtained from $f(x)=\sin\left(\pi x/2\right)$. The fixation vector is given by $F_i=f(i/N)$ with $i=0,\dots,N$, and a population of $N=100$ individuals. The shadow represents the tolerable error, i.e, the region between $f(x)\pm0.01$. The fixation probability for 2-player game with pay-off vectors $\mathbf{a}^{(2)}\approx(252,   0.959)$ and $\mathbf{b}^{(2)}\approx(254,   2.16)$ is sobreposed to the fixation function and error bar (grey area).
Right: $f(x)=1-\sqrt{1-x}$, $N=100$. We find that $d=2$ provides a good approximation, with $\mathbf{a}^{(2)}\approx (2.74,-0.00870)$ and $\mathbf{b}^{(2)}\approx (2.77,0.0270)$, i.e., $d_{\min}=2$.
In the former case, \A dominates \B while the reverse occurs in the latter.}
\label{fig:dominance}
\end{figure}

\begin{figure}
\centering
\includegraphics[width=0.47\textwidth]{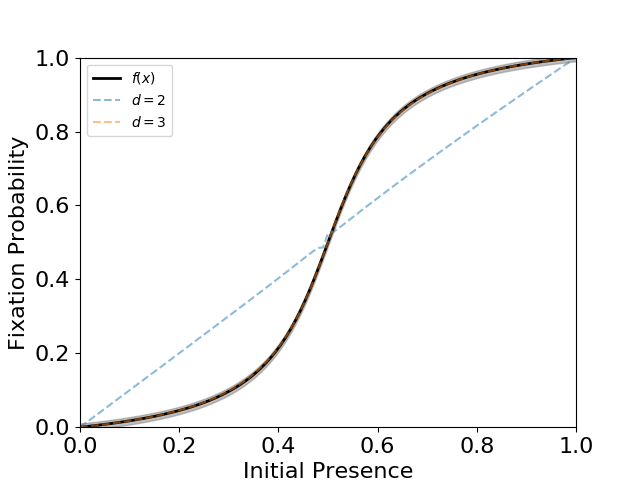}
  \includegraphics[width=0.47\textwidth]{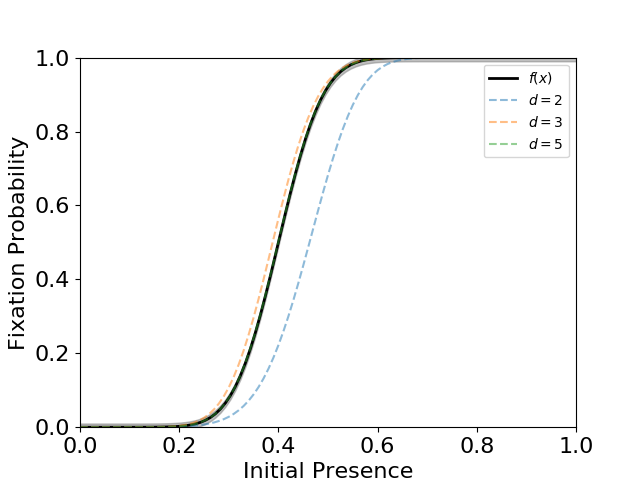}
 \caption{Coordination case, always with $N=100$ individuals. Left: $f(x)=(\arctan(10x-5)+\arctan(5))/(2\arctan(5))$. In this case, the approximation with $d=2$ clearly fails to provide a good approximation  to the fixation probability. The first good approximation is given by $d_{\min}=3$ with, $\mathbf{a}^{(3)}\approx( 4.95, -4.46,  5.23)$ and
$\mathbf{b}^{(3)}\approx( 5.25, -4.46,4.93)$.
Right: $f(x)=(\erf(10x-4)-\erf(-4))/(\erf(6)-\erf(-4))$ with $\erf(x)=\frac{2}{\sqrt{\pi}}\int_0^x\e^{-z^2}\rd z$. The first good approximation is given by $d_{\min}=5$, $\mathbf{a}^{(5)}\approx(16.6,  13.3,  3.54, -0.217, 0.142)$ and $\mathbf{b}^{(5)}\approx( 25.5,  12.3,  2.01, -0.0532,
0.00235)$.}
 \label{fig:coexistence}
 \end{figure}

\begin{figure}
\centering
 \includegraphics[width=0.47\textwidth]{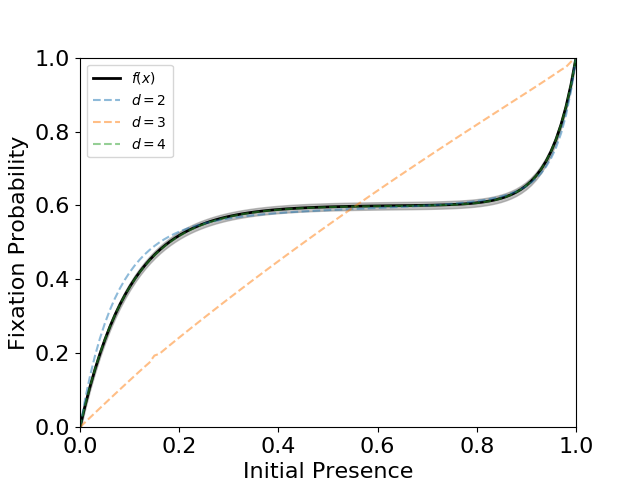}
 \includegraphics[width=0.47\textwidth]{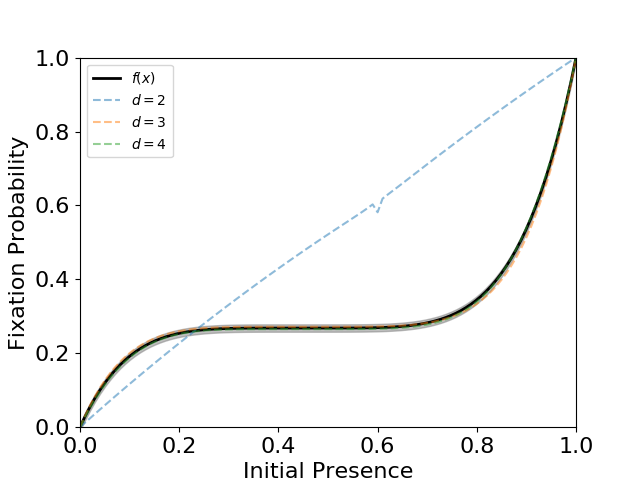}
 \caption{Coexistence case, always with $N=100$ individuals. Left: $f(x)=\frac{4}{10}\frac{1-\e^{20x}}{1-\e^{20}}+\frac{6}{10}\frac{1-\e^{-10x}}{1-\e^{-10}}$ and $N=100$. The first good approximation is given by $d_{\min}=4$, with $\mathbf{a}^{(4)}\approx(1.30\times 10^3,-157,-59.0,172)$ and $\mathbf{b}^{(4)}\approx(1.27\times 10^3,-131,-73.5,182)$. Right: $f(x)=((x-9/20)^5+(9/20)^5)/((11^5+9^5)/20^5)$. The first good approximation is given by $d_{\min}=4$ and $\mathbf{a}^{(4)}\approx(0.799, -0.349, -0.230,  1.25)$ and $\mathbf{b}^{(4)}\approx( 0.804, -0.366, -0.211,  1.24)$.}
\label{fig:coordination}
 \end{figure}

As discussed in Section~\ref{sec:prelim}, if $f$ is strictly increasing then $\Upsilon_{\bF}$ is monotonically increasing. Nevertheless, even in this case we might have quite complex fixation patterns as showed in \citep{ChalubSouza16,ChalubSouza18} and hence its numerical inversion may not be straightforward. However, if $f'(x)\geq \alpha=\mathcal{O}(1)$, then  the typical corresponding fitness patterns are not too complex --- thus, games with a small number of players will suffice.  When $f$ is no longer increasing, or if $f$ has plateau-like behaviour in some regions, then the numerical inversion of $\Upsilon_{\bF}$ is more involved. In addition, the fitness patterns are also typically more complex, and thus we will need games with a larger number of players to fit the corresponding fitness landscape. Fig.~\ref{fig:decreasing} shows that this might happen, even when there is a slightly deviation from monotonicity.

\begin{figure}
 \centering
 \includegraphics[width=0.47\textwidth]{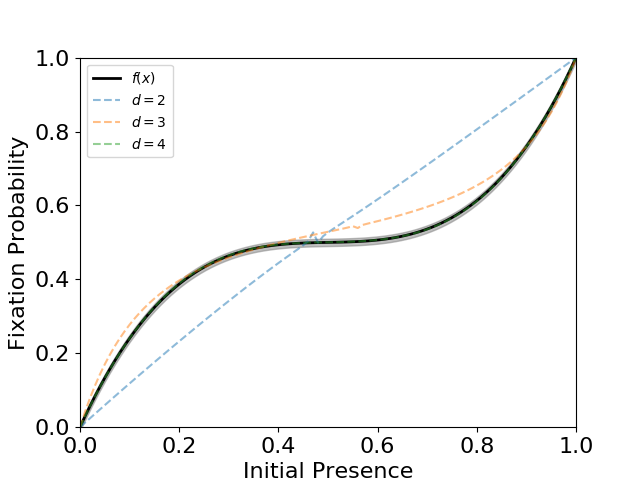}
 \includegraphics[width=0.47\textwidth]{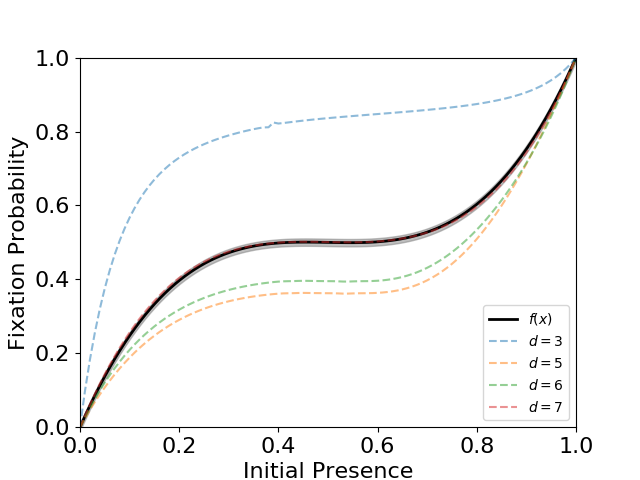}
 \caption{$f(x)=(3.9)x(1/2-x)(1-x)+x$ (left) and $f(x)=4.1x(1/2-x)(1-x)+x$ (right), with $N=100$. Despite the fact that in the right case the fixation function is only slightly decreasing in the interval $x\in(0.455,0.546)$, $d_{\min}=7$, while in the left case, where $f$ is increasing, implies that $d_{\min}=4$. Namely 
 $\mathbf{a}^{(4)}\approx  (3.35, -1.13, -0.764,  3.05)$
 and $\mathbf{b}^{(4)}\approx(3.42, -1.12, -0.775,   2.99)$ (left) and
 $\mathbf{a}^{(7)}\approx ( 1.21, -1.14,  1.10, -1.10,  1.14, -1.22,  1.35)$ and $\mathbf{b}^{(7)}\approx
( 1.34, -1.21,  1.14, -1.10,  1.11, -1.15, 1.23)$ (right). Note that a small change in the fixation function may change completely the game, in special if the function is not strictly increasing any more.}
 \label{fig:decreasing}
\end{figure}

\section{Large population limits}

In all simulations presented so far, we considered a population of $N=100$ individuals. Here, we will briefly discuss the relation between the minimum game $d_{\min}$ and the population size.

On one hand, it is clear that we may design a WF process with pay-off given by $d_0$-player game, with $d_0$ fixed irrespective of $N$. With the right set of assumptions for the pay-off dependence on $N$, it is possible to obtain a diffusion approximation and, therefore, a fixation probability. The inverse problem will then find a minimum game $d_{\min}\le d_0$, depending on the tolerable error, i.e., to provide a simplification -- an order reduction -- of the original game.

On the other hand, it is not difficult to design a family of games in which the order of the game depends on the population size. Namely, assume that the focal type has a large probability to survive if and only if its presence is larger than a certain critical fraction of the population. Assume that $p_i\approx 0$ for $i<x_{\mathrm{c}}N$ and $p_i\approx 1$ for $i>x_{\mathrm{c}}N$, with $x_{\mathrm{c}}\in(0,1)$. In other words, \A and \B collaborate with individuals of the same type, and in large groups. If the fraction of \A individuals is larger than $x_{\mathrm{c}}$, then \A prevails, otherwise \B prevails. The fixation function will then  be approximately given by $\theta(x-x_{\mathrm{c}})$, where $\theta$ is the Heaviside function. We expect that the minimum order of the game able to reproduce this fixation probability will increase in $N$. However, from the computational point of view, the required computing time seems to increase exponentially with $d$ and, therefore, we are not able to increase the order of the game arbitrarily. 

These two examples show that, for certain fixation functions, it is natural to expect  $d_{\min}$ to increase as a function of $N$ up to a certain bound, while for other fixation functions $d_{\min}$ may increase without bound. Considering that real populations can be extremely large, it is important to know \emph{a priori} in which case we are. 

At this point, we are far from providing any explanation in the behaviour of $d_{\min}$ when $N\to\infty$. We will, however, present some simulations that illustrate how the maximum error between the original fixation probability $F_i$ and the computed fixation probability $\tilde F_i$, for a given $d$, depends on $N$ in Fig.~\ref{fig:bigN}.

\begin{figure}
	\centering
	\subfloat[red][$f(x)=\sin(\pi x/2)$]{\includegraphics[width=0.48\textwidth]{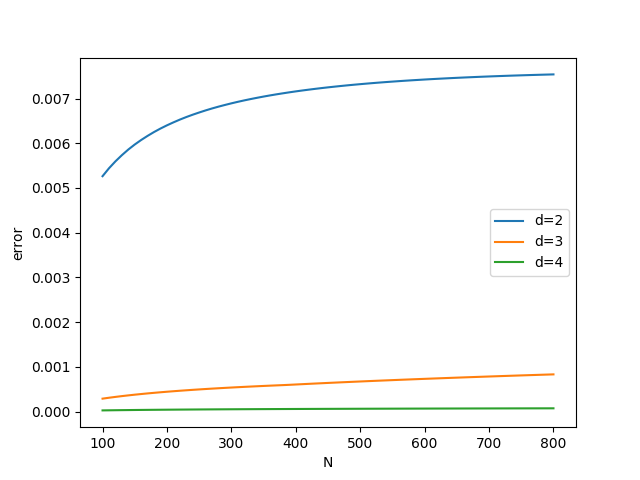}}
	\subfloat[darkgreen][$f(x)=1-\sqrt{1-x}$]{\includegraphics[width=0.48\textwidth]{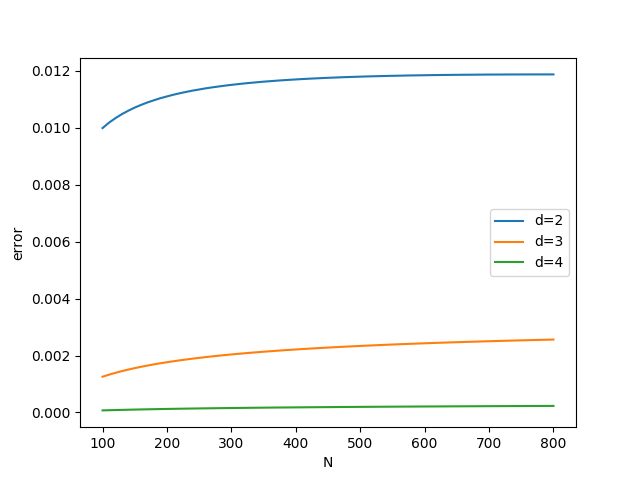}}\\
	\subfloat[blue][$f(x)=\frac{\arctan(10x-5)+\arctan(5)}{2\arctan(5)}$]{\includegraphics[width=0.48\textwidth]{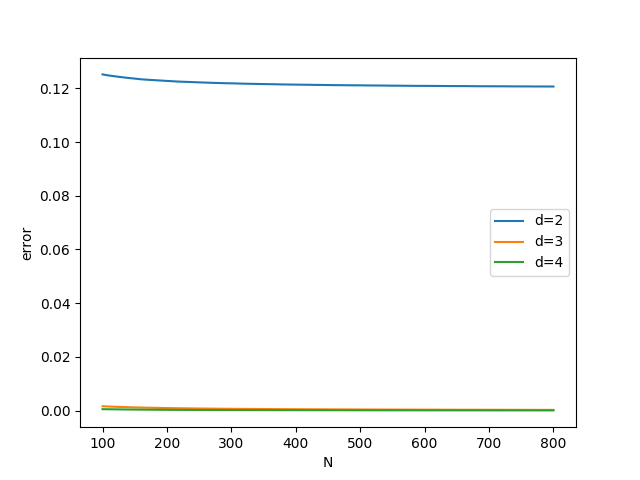}}
	\subfloat[yellow][$f(x)=\frac{4}{10}\frac{1-\e^{20x}}{1-\e^{20}}+\frac{6}{10}\frac{1-\e^{-10x}}{1-\e^{-10}}$]{\includegraphics[width=0.48\textwidth]{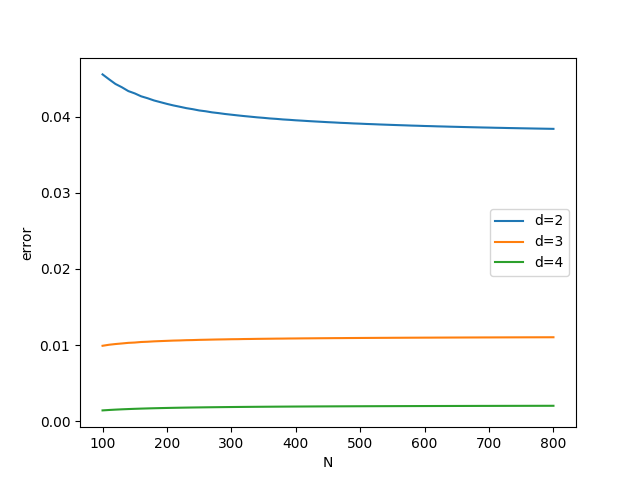}}\\
	\subfloat[cyan][$f(x)=3.9x(1/2 -x)(1 - x) + x$]{\includegraphics[width=0.48\textwidth]{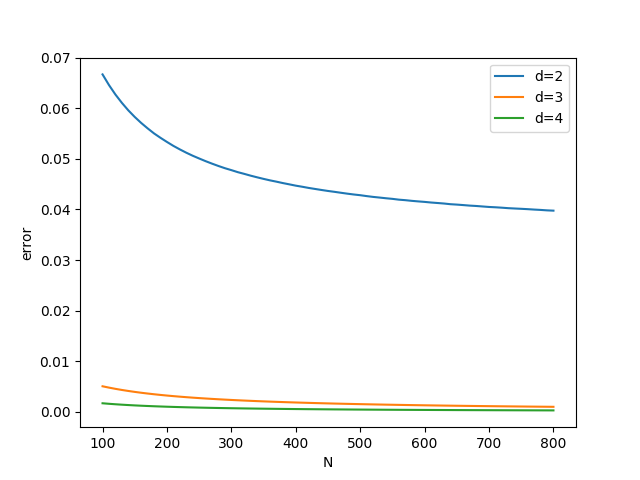}}
	\subfloat[pink][$f(x)=4.1x(1/2 -x)(1 - x) + x$]{\includegraphics[width=0.48\textwidth]{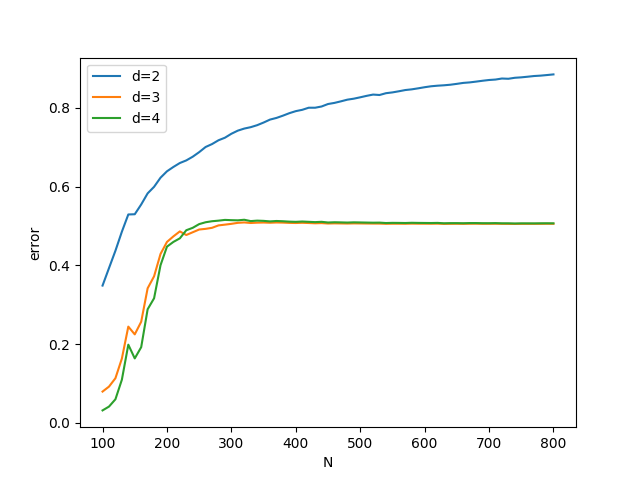}}
	\caption{Dependence of $\mathcal{E}(d,N)=\max_i|\tilde F_i-F_i|$ on $N$ for some values of $d$. For the moment, we are not able to provide a general picture of the function $\mathcal{E}$, however from our numerical experiments, we notice that its dependence on $N$ is small, albeit extremely complicated, with the important exception of the case in which $f$ is not monotonically increasing (example~(f)).}
	\label{fig:bigN}
\end{figure}

\section{Discussion}

Fitness is one of the central concepts in the mathematical description of evolutionary processes. Despite the fact that in most models, a class of simple fitness function (e.g. constant or affine) is used, there are no a priori limitations on the   functions that shall be used when modelling real problems. However, a possible fitness function, obtained say, from experiments, provides few, if any, insight in the dynamics between individuals in the population under study.

The aim of this work is to provide a link between macroscopic observables (as the fixation probability) and microscopic, first-principles, description of the population (\emph{finite population models}). Therefore, we show that any fixation pattern can be obtained as a result of the Wright-Fisher process with explicitly calculated interaction between individuals. The more precision we want in the macroscopic description, the more complex will be the game compatible with the data. However, an interesting insight in the internal modus operandi of a population can be obtained from the game that is not evident in the fixation pattern, or even in the fitness function.

In order to achieve that goal, we continued previous works from the authors  in which it is shown that any  fixation pattern in a finite population can be realised by a Wright-Fisher (WF) model. More precisely, it was shown in~\citep{ChalubSouza:2017a} that given a $N+1$  vector $\bF$ with $F_0=0$, $F_N=1$ and $0<F_i<1$, for $i=1,\ldots,N-1$, then there exists a choice of Type Selection Probabilities (TSPs) for which the corresponding WF model has exactly $F$ as its fixation vector.

One clear limitation of the present work, that deserves further investigation, is that we did not study the large population limit of games, when $N\to\infty$. From the numerical experiments, we gathered evidence that certain fixation patterns can be reproduced by game with a finite number of players, even if the population is infinite, whereas other patterns will require a number of players that increases along with population size.

We finish by noticing that EGT has become pervasive in multiple applications --- notably in the mathematical study of biological and social evolution. However, EGT classical framework is infinite, deterministic, and  well mixed populations. This work shows how it can be extended to describe the dynamics of finite populations with demographic noise --- although still well mixed and fixed size populations. Further development is possible along many avenues: including structure, adding heterogeneity, or even moving towards individual based models. Certainly, EGT has still a lot to offer.

\section*{Acknowledgements}
FACCC was partially supported by FCT/Portugal Strategic Project UID/MAT/00297/2013 (Centro de Matemática e Aplicações, Universidade Nova de Lisboa) and by a ``Investigador FCT'' grant. 
MOS was partially supported by CNPq under grants   \# 486395/2013-8 and  \# 309079/2015-2. We also thank the useful comments by two anonymous reviewers, which helped to improve the manuscript.

\bibliographystyle{abbrvnat}
\bibliography{pade}

\end{document}